\documentclass[twocolumn,showpacs,laps,prl,amsmath,amssymb]{revtex4}

\usepackage{bm}
\usepackage{graphicx,color}
\usepackage{amsthm}

\usepackage{xy}
\xyoption{matrix}
\xyoption{frame}
\xyoption{arrow}
\xyoption{arc}
\usepackage{ifpdf}
\ifpdf
\else
\PackageWarningNoLine{Qcircuit}{Qcircuit is loading in Postscript mode.  The Xy-pic options ps and dvips will be loaded.  If you wish to use other Postscript drivers for Xy-pic, you must modify the code in Qcircuit.tex}
\xyoption{ps}
\xyoption{dvips}
\fi
\entrymodifiers={!C\entrybox}

\newcommand{\ket}[1]{{\left\vert{#1}\right\rangle}}
\newcommand{\qw}[1][-1]{\ar @{-} [0,#1]}
\newcommand{\qwx}[1][-1]{\ar @{-} [#1,0]}
\newcommand{\cw}[1][-1]{\ar @{=} [0,#1]}

\newcommand{\gate}[1]{*+<.6em>{#1} \POS ="i","i"+UR;"i"+UL **\dir{-};"i"+DL **\dir{-};"i"+DR **\dir{-};"i"+UR **\dir{-},"i" \qw}
\newcommand{\meter}{*=<1.8em,1.4em>{\xy ="j","j"-<.778em,.322em>;{"j"+<.778em,-.322em> \ellipse ur,_{}},"j"-<0em,.4em>;p+<.5em,.9em> **\dir{-},"j"+<2.2em,2.2em>*{},"j"-<2.2em,2.2em>*{} \endxy} \POS ="i","i"+UR;"i"+UL **\dir{-};"i"+DL **\dir{-};"i"+DR **\dir{-};"i"+UR **\dir{-},"i" \qw}

\newcommand{\multimeasure}[2]{*+<1em,.9em>{\hphantom{#2}} \qw \POS[0,0].[#1,0];p !C *{#2},p \drop\frm<.9em>{-}}

\newcommand{\multigate}[2]{*+<1em,.9em>{\hphantom{#2}} \POS [0,0]="i",[0,0].[#1,0]="e",!C *{#2},"e"+UR;"e"+UL **\dir{-};"e"+DL **\dir{-};"e"+DR **\dir{-};"e"+UR **\dir{-},"i" \qw}
\newcommand{\ghost}[1]{*+<1em,.9em>{\hphantom{#1}} \qw}

\newcommand{\rstick}[1]{*!L!<-.5em,0em>=<0em>{#1}}
\newcommand{\lstick}[1]{*!R!<.5em,0em>=<0em>{#1}}

\newcommand{\Qcircuit}{\xymatrix @*=<0em>}

\newcommand{\nn}{\nonumber \\}
\newcommand{\kets}[1]{ |{#1} \rangle}
\newcommand{\herr}{\gamma}
\newcommand{\segs}{{ T}}
\newcommand{\mE}{{\mathbb E}}

\newtheorem{theorem}{Theorem}
\newtheorem{lemma}{Lemma}

\begin{document}
\title{Exponential improvement in precision for Hamiltonian-evolution simulation}
\author{D. W. Berry${}^1$, R. Cleve${}^2$, and R. D. Somma${}^3$}
\affiliation{${}^1$Department of Physics and Astronomy, Macquarie University, Sydney, NSW 2109, Australia\\
${}^2$Institute for Quantum Computing and School of Computer Science, University of Waterloo, ON N2L 3G1, Canada \\
${}^3$Theory Division, Los Alamos National Laboratory, Los Alamos, NM 87545, USA}
\date{\today}

\begin{abstract}
We provide a quantum method for simulating Hamiltonian evolution with complexity
polynomial in the logarithm of the inverse error. This is an exponential improvement over existing methods
for Hamiltonian simulation. In addition, its scaling with respect to time is close to linear, and its scaling
with respect to the time derivative of the Hamiltonian is logarithmic. These scalings improve upon most
existing methods. Our method is to use a compressed Lie-Trotter formula, based on recent ideas for efficient
discrete-time simulations of continuous-time quantum query algorithms.
\end{abstract}

\maketitle

The simulation of 
physical quantum systems is important to understand novel physical phenomena and is a major potential application of quantum computers.
This is the original reason why Feynman proposed the concept of a quantum computer in 1982 \cite{Fey82}.
Typically the simulation required is the time evolution of the quantum state under a Hamiltonian.
In 1996, Lloyd \cite{Lloyd96} showed how to explicitly simulate evolution under physical Hamiltonians
that are expressible as a sum of local interaction Hamiltonians,
with scaling polynomial with respect to the number of qubits and the evolution time.

Aharonov and Ta-Shma \cite{Aharonov03} formulated a more abstract form of Hamiltonian simulation, dubbed the \textit{sparse Hamiltonian} problem.
Here, a sparse Hamiltonian (i.e.\ with a bounded number of nonzero entries in each row/column) is specified by an oracle for the values and positions of its nonzero entries.
Sparse Hamiltonians include the physical case of sums of local interaction Hamiltonians.
In addition to serving as a generic formulation for the problem of simulating Hamiltonian evolution, sparse Hamiltonian simulation
provides a method to design new quantum algorithms~\cite{Childs09b,Harrow09}. 
Beginning with~\cite{Aharonov03}, a series of simulations were discovered~\cite{Childs04,Berry2007,Childs08,Wiebe10,Wiebe11,Childs11,BerryC12} with improvements in efficiency in various parameters.

A highly desirable feature of algorithms is that doubling the number of significant digits 
only increases the complexity by a constant factor.
This will be true if the complexity is polynomial in $\log(1/\varepsilon)$, where $\varepsilon$ is the allowable error.
A drawback to previous quantum algorithms for Hamiltonian simulation is that their complexity is significantly worse than this \cite{Brown06,Brown10}.
Another drawback is that, for time-dependent Hamiltonians, the complexity depends strongly on the time-derivatives of the Hamiltonian \cite{Wiebe10,Wiebe11}.
One technique avoids that dependence, at the expense of worse scaling in both the error and the evolution time \cite{Poulin11}.

In this work we present a new approach to achieve Hamiltonian simulations with complexity scaling polynomial
in $\log(1/\varepsilon)$; this is an exponential improvement over prior approaches.
Our algorithm also provides scaling in the norm of the time derivative of the Hamiltonian, $\|H'\|$, evolution time and sparseness that is improved over
that from prior work in~\cite{Childs04,Berry2007,Childs11,Wiebe10,Wiebe11}.
The basis of our approach is to decompose sparse Hamiltonians into sums of self-inverse Hamiltonians,
then extend the methodologies in~\cite{CleveG+2009,Berry12} to apply to these sums.
We present the result in terms of sparse Hamiltonians for generality,
but it equally holds for sums of interaction Hamiltonians.

Our result for the complexity in terms of the calls to the oracle plus additional gates is as follows.
\begin{theorem}
\label{thm:main}
A Hamiltonian $H$ acting on $n$ qubits with sparseness $d$ (at most $d$ nonzero entries in each row/column) and an oracle for the values and
positions of the nonzero elements can be simulated with accuracy $\varepsilon$ for time $t$ with 
\begin{equation}
\label{eq:addgat}
O\left([d^2 \tau+\log(1/\varepsilon)] \log^3[d(\tau+\tau')/\varepsilon]n^c\right)
\end{equation}
calls to the oracle and additional (one- and two-qubit) gates, where $\tau := \|H\| t$, $\tau' := \|H'\| t$,
and $c$ is a constant.
\end{theorem}

The factors that depend on $n$ are similar to previous work~\cite{CleveG+2009,Berry12}, and will not be discussed further.
Units are used such that $\hbar=1$.
We take $\|H\|$ and $\|H'\|$ to be the maximum of the spectral norm over the interval $[0,t]$.

The number of both oracle calls and additional gates scales linearly in
$d^2\|H\| t$, times a factor that is logarithmic in the system parameters.
Most significantly, the complexity scales polynomially in $\log(1/\varepsilon)$,
which is a dramatic improvement over previous methods.
In addition, the complexity scales logarithmically in the norm of the derivative of $H$,
whereas most previous techniques scaled polynomially in $\|H'\|$.
The scaling in $t$ is only a logarithmic factor away from being linear, which is
better than any previous scheme based on Lie-Trotter formulae.
One technique provides scaling strictly linear in the time, at the expense of worse scaling in the error and not allowing time-dependence \cite{Childs08,BerryC12}.

The oracle for the Hamiltonian is similar to that in~\cite{Berry2007}.
That is, given any column $j$ and a value for time $t'\in[0, t]$, there is a function giving the position and value of the $i$'th nonzero element in the $j$'th column of $H(t')$.
Intuitively, the oracle can be interpreted as a subroutine for computing the nonzero elements of the Hamiltonian.

Reference~\cite{CleveG+2009} considers the sum of two Hamiltonians: a ``query'' Hamiltonian, and a driving Hamiltonian.
The crucial feature is that the query Hamiltonian is self-inverse (up to a trivial change in scaling).
In the present context there is no driving Hamiltonian;
the Hamiltonian to be simulated (i.e., $H(t')$) need not be self-inverse, but it
can be decomposed into a sum of self-inverse Hamiltonians as described below.
Then, we can generalize the methods in~\cite{CleveG+2009,Berry12} to apply to sums of
many self-inverse Hamiltonians.
To ensure the scaling is polynomial in $\log(1/\varepsilon)$, we improve the error analysis of~\cite{CleveG+2009,Berry12}.

The design of the algorithm has the following basic steps.
Details on these steps will follow.
\begin{enumerate}
\item The Hamiltonian is decomposed into a sum of real or imaginary 1-sparse Hamiltonians $H_j$~\cite{Berry2007}.
\item The Hamiltonian evolution is approximated by a Lie-Trotter product formula of evolutions under the $H_j$.
\item The $H_j$ are decomposed into sums of commuting self-inverse Hamiltonians $U_l$.
\item The evolution under each $U_l$ for a short time interval is implemented with high probability
using a control qubit, a controlled application of $U_l$, and a measurement on the control qubit.
\item The evolution is broken up into a number of \textit{segments}, which are contiguous blocks of short time intervals.
For each segment, a concentration bound and a compression technique enable the number of applications of the controlled-$U_l$ operations,
or calls to the oracle for $H$, to be limited.
\item If the desired measurement result is not obtained on all control qubits, then there is an error which can be corrected with high probability via a recursive procedure.
\end{enumerate}

For sparse Hamiltonians, step 1 can be applied using previous techniques such as in~\cite{Berry2007}.
For the specific case of a Hamiltonian that is a sum of interaction Hamiltonians, we can use this sum as an initial decomposition.
Then each interaction Hamiltonian is sparse, and we can compute the positions and values of its entries.
This provides an oracle for using techniques such as in~\cite{Berry2007} to decompose into 1-sparse Hamiltonians.

We can easily separate these 1-sparse Hamiltonians into real  and imaginary parts.
The reason why we then consider the Lie-Trotter product formula (step 2) rather than immediately decomposing into self-inverse Hamiltonians (step 3),
is that the error bound for the Lie-Trotter product formula relies on the Hamiltonian varying continuously.
After using the Lie-Trotter product formula, we can decompose the 1-sparse Hamiltonians into self-inverse Hamiltonians which may vary discretely.
Because these are mutually commuting, the exponential of the sum is exactly decomposed as a product of exponentials with no additional error.
The decomposition into self-inverse Hamiltonians is as in the following Lemma.

\begin{lemma}
\label{lem:dec}
A 1-sparse real or imaginary Hamiltonian $G$ can be approximated within max-norm distance $\herr$ via an
equally weighted sum of $O(\|G\|_{\mbox{\rm\scriptsize max}}/\herr)$ commuting self-inverse Hamiltonians.
\end{lemma}

\begin{proof}
A real 1-sparse Hamiltonian can be expressed as a linear combination of two commuting 1-sparse Hamiltonians, where one is diagonal, and the other is off-diagonal.
First consider the off-diagonal component, and denote it as $\widetilde{G}$.
Since $\widetilde{G}$ is 1-sparse, it is a direct sum of $2 \times 2$ blocks, each of which is a real multiple of the Pauli sigma matrix $\sigma_x$.
These multiples are in the interval $[-g,+g]$, where 
$g = \|G\|_{\mbox{\rm\scriptsize max}}$.
Therefore, they can each be rounded to the nearest rational multiple of $g$,
which is of the form $(h/\ell)g$, where 
$h \in \{-\ell,\dots,+\ell\}$, and the resulting error is at most $g/(2\ell)$ (the value to set $\ell$ will be discussed later).
This rounded $\widetilde{G}$ can be expressed as 
$(g/\ell) \sum_{j=1}^{\ell} G_j$, 
where each $G_j$ is 1-sparse with entries in $\{-1,0,+1\}$.
A small illustrative example (with $\ell=2$) is
\begin{align}\label{eq:ex1}
\begin{pmatrix}
0 & 1 & 0 & 0 \\
1 & 0 & 0 & 0 \\
0 & 0 & 0 & \frac{1}{2} \\
0 & 0 & \frac{1}{2} & 0
\end{pmatrix} 
=
\frac 12 
\begin{pmatrix}
0 & 1 & 0 & 0 \\
1 & 0 & 0 & 0 \\
0 & 0 & 0 & 1 \\
0 & 0 & 1 & 0
\end{pmatrix} 
+ 
\frac 12 
\begin{pmatrix}
0 & 1 & 0 & 0 \\
1 & 0 & 0 & 0 \\
0 & 0 & 0 & 0 \\
0 & 0 & 0 & 0
\end{pmatrix}.
\end{align}
Formally, we can define the matrix entries of $G_j$ as 
\begin{align}
G_j[a,b] = 
\begin{cases}
+1 & \mbox{if $0 \le j < {\rm round}(\ell\widetilde{G}[a,b]/g)$} \\
-1 & \mbox{if ${\rm round}(\ell\widetilde{G}[a,b]/g) < j \le 0$} \\
\phantom{+}0 & \mbox{otherwise},
\end{cases}
\end{align}
from which an oracle for each $G_j$ can be implemented in terms of the oracle for $G$.

The Hamiltonians $G_j$ can have zero eigenvalues, but these can be eliminated by one further stage of splitting into two self-inverse Hamiltonians.
For example, 
\begin{align}
\begin{pmatrix}
0 & 1 & 0 & 0 \\
1 & 0 & 0 & 0 \\
0 & 0 & 0 & 0 \\
0 & 0 & 0 & 0
\end{pmatrix} 
= 
\frac 12 
\begin{pmatrix}
0 & 1 & 0 & 0 \\
1 & 0 & 0 & 0 \\
0 & 0 & 1 & 0 \\
0 & 0 & 0 & 1
\end{pmatrix} 
+ 
\frac 12 
\begin{pmatrix}
0 & 1 & \phantom{\!-}0 & \phantom{\!-}0 \\
1 & 0 & \phantom{\!-}0 & \phantom{\!-}0 \\
0 & 0 & \!-1 & \phantom{\!-}0 \\
0 & 0 & \phantom{\!-}0 & \!-1
\end{pmatrix}
\end{align}
for the second Hamiltonian in Eq.~(\ref{eq:ex1}).
We have only discussed the real off-diagonal component, but a similar procedure exists for the diagonal component (with $1 \times 1$ real-valued blocks),
as well as for the off-diagonal imaginary case (with $2 \times 2$ $\sigma_y$ blocks, instead of $\sigma_x$).
The reason why we consider real or imaginary Hamiltonians for this Lemma is that real and imaginary components in the same block would not commute.
Restricting the Hamiltonian to be real or imaginary ensures that all self-inverse components commute.
Setting $\ell = O(g / \herr)$ leads to the required precision (max-norm distance).
\end{proof}

As a result of Lemma~\ref{lem:dec}, 
the Hamiltonian can be approximated as an equally-weighted sum
of self-inverse Hamiltonians $U_l$; that is, $H\approx \herr \sum_l U_l$.
Then, the methods in~\cite{CleveG+2009,Berry12} can be extended to simulate the evolution under $H$.
Because the Lie-Trotter formula results in short-time evolutions under the $U_l$, we have
\begin{equation}
e^{-iU_l s} = \openone \cos s - i U_l \sin s,
\end{equation}
where the time interval is $\delta t=s/\gamma$.
In~\cite{CleveG+2009} we provided an algorithm to implement $e^{-iU_l s}$ that uses an ancilla qubit,
implements the operation $U_l$ controlled on the state $\ket 1$ of the ancilla,
and performs a measurement (step 4). If the result of the measurement is $b=0$, the simulation succeeds. 
The initial state of the ancilla is $\ket 0$ and is acted on by the one-qubit unitaries $R$ and $P$,
with
\begin{equation}\label{eq:rdef}
R := \frac 1{\nu}
\left(
\begin{array}{lr}
\sqrt{\cos s} & \sqrt{\sin s} \\
\sqrt{ \sin s} & -\sqrt{\cos s}
\end{array}
\right),
\ \ \ \ \
P :=
\left(
\begin{array}{lr}
1 & 0 \\
0 & -i
\end{array}
\right),
\end{equation}
where $\nu^2={|\cos s| + |\sin s|}$ is for normalization.
We can then simulate each term of the Lie-Trotter product formula using one of these ancillas
and a controlled-$U_l$ operation, which can be implemented with a call to the oracle for $H$. While 
 the number of oracle calls with this method is still too large (it is the number of terms in the 
Lie-Trotter formula), it can be reduced significantly using the results in~\cite{CleveG+2009,Berry12}, as discussed below.

Next, we divide the time interval $[0,t]$ into smaller intervals or segments.
The length of each segment is chosen such that there is a probability $\le 1/4$ of failure (the wrong measurement result on any of the control qubits for that segment).
The state of the ancillas for a segment is a combination of bit strings,
where the positions of the 1's implicitly control $U_l$.
We can implement the same operation in an alternative way, by performing $U_l$
explicitly controlled on the state of all qubits (Fig.~\ref{fig:unco}).
The first controlled operation selects the value of $l$ in $U_l$ based on the position of the first 1 in the string,
the second controlled operation uses the position of the second 1, and so forth. If there are no further 1's, the controlled operation simply performs the identity.
Since $s \ll 1$, the state of the control qubits in the segment will have most of its weight on bit strings with low Hamming weight.
We can therefore truncate the Hamming weight at $k_1$, and limit the resulting error to $\varepsilon_1$,
which is exponentially small. Using the simulation in Fig.~\ref{fig:unco},
the number of oracle calls in the segment is reduced to $k_1$.
The choice of $\varepsilon_1$ and $k_1$ will yield the oracle complexity.
To reduce the number of gates or qubits needed in each segment, we use
the compression technique of~\cite{Berry12} almost directly.

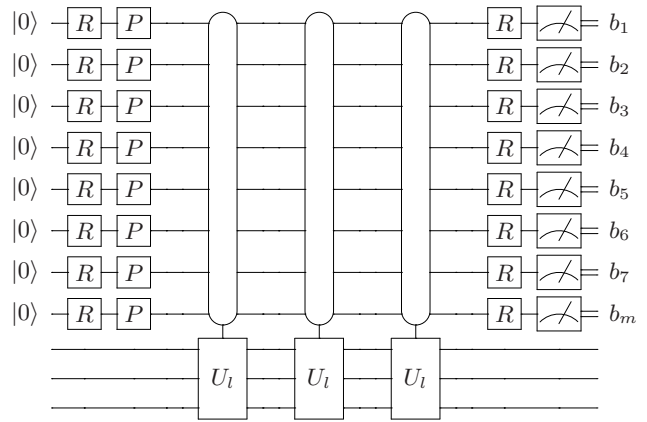
\begin{figure}
\centerline{\Qcircuit @C=0.65em @R=0.3em {
& \lstick{\kets{0}} & \gate{R} & \gate{P} & \qw & \qw & \multimeasure{7}{\,} \qwx[8] & \qw & \qw & \multimeasure{7}{\,} \qwx[8] & \qw & \qw & \multimeasure{7}{\,} \qwx[8] & \qw & \qw & \gate{R} & \meter & \rstick{b_1} \cw \\
& \lstick{\kets{0}} & \gate{R} & \gate{P} & \qw & \qw & \ghost{\,} & \qw & \qw & \ghost{\,} & \qw & \qw & \ghost{\,} & \qw & \qw & \gate{R} & \meter & \rstick{b_2} \cw \\
& \lstick{\kets{0}} & \gate{R} & \gate{P} & \qw & \qw & \ghost{\,} & \qw & \qw & \ghost{\,} & \qw & \qw & \ghost{\,} & \qw & \qw & \gate{R} & \meter & \rstick{b_3} \cw \\
& \lstick{\kets{0}} & \gate{R} & \gate{P} & \qw & \qw & \ghost{\,} & \qw & \qw & \ghost{\,} & \qw & \qw & \ghost{\,} & \qw & \qw & \gate{R} & \meter & \rstick{b_4} \cw \\
& \lstick{\kets{0}} & \gate{R} & \gate{P} & \qw & \qw & \ghost{\,} & \qw & \qw & \ghost{\,} & \qw & \qw & \ghost{\,} & \qw & \qw & \gate{R} & \meter & \rstick{b_5} \cw \\
& \lstick{\kets{0}} & \gate{R} & \gate{P} & \qw & \qw & \ghost{\,} & \qw & \qw & \ghost{\,} & \qw & \qw & \ghost{\,} & \qw & \qw & \gate{R} & \meter & \rstick{b_6} \cw \\
& \lstick{\kets{0}} & \gate{R} & \gate{P} & \qw & \qw & \ghost{\,} & \qw & \qw & \ghost{\,} & \qw & \qw & \ghost{\,} & \qw & \qw & \gate{R} & \meter & \rstick{b_7} \cw \\
& \lstick{\kets{0}} & \gate{R} & \gate{P} & \qw & \qw & \ghost{\,} & \qw & \qw & \ghost{\,} & \qw & \qw & \ghost{\,} & \qw & \qw & \gate{R} & \meter & \rstick{b_m} \cw \\
& & \qw & \qw & \qw & \qw & \multigate{2}{U_l} & \qw & \qw & \multigate{2}{U_l} & \qw & \qw & \multigate{2}{U_l} & \qw & \qw &\qw & \qw & \qw\\
& & \qw & \qw & \qw & \qw & \ghost{U_l} & \qw & \qw & \ghost{U_l} & \qw & \qw & \ghost{U_l} & \qw & \qw &\qw & \qw & \qw\\
& & \qw & \qw & \qw & \qw & \ghost{U_l} & \qw & \qw & \ghost{U_l} & \qw & \qw & \ghost{U_l} & \qw & \qw &\qw & \qw & \qw}}
\caption{\label{fig:unco}The Lie-Trotter formula implemented with control qubits for a single segment.
The self-inverse Hamiltonians $U_l$ are controlled by the positions of the 1's in the control qubits.}
\end{figure}

We note that the approach in this paper differs from that in~\cite{CleveG+2009,Berry12}, where the evolution under a driving
Hamiltonian is applied for a time controlled by the positions of the 1's in the control qubits. In particular,
the current approach allows us to simulate a sum of many self-inverse Hamiltonians $U_l$,
where each $U_l$ can be implemented using the oracle for $H$ via standard techniques.

We now use these methods to prove Theorem \ref{thm:main}.
\begin{proof}[Proof of Theorem \ref{thm:main}.]
In step 1, for any $t' \in [0,t]$, we decompose the Hamiltonian $H(t')$ into $O(d^2)$ 1-sparse Hamiltonians $H_j(t')$
using techniques such as in~\cite{Berry2007}.
It is trivial to further divide these 1-sparse Hamiltonians into real and imaginary parts.
We use the lowest-order Lie-Trotter formula for the time-ordered exponential
for a short time period $[t_0,t_0+\delta t]$
\begin{equation}
{\cal T} \exp\left[\int_{t_0}^{t_0+\delta t} iH(t') dt' \right] \approx \prod_{j=1}^M \exp\left[ iH_j(t_0) \delta t \right],
\end{equation}
where ${\cal T}$ indicates time-ordering, and $M=O(d^2)$ is the number of Hamiltonians in the decomposition.
The error in this approximation is $O(M^2 \delta t^2(\|H\|^2+\|H'\|))$.
To obtain error $O(\varepsilon)$, we divide the long time interval $[0,t]$ into $r$ smaller intervals of length $\delta t=t/r$ with
\begin{equation}
r=\Theta \left( (Mt)^2(\|H\|^2+\|H'\|) /\varepsilon \right).
\end{equation}
The number of exponentials of the form $\exp[-iH_j(t')\delta t]$ in the approximation is $rM$.

In step 3, we use Lemma \ref{lem:dec} to approximately decompose the Hamiltonians $H_j$ into a sum of $O(\|H\|_{\rm max}/\herr)$ commuting and self-inverse terms $U_l$.
The exponential of the sum is written as a product of exponentials of $U_l$ without introducing additional error.
The result is an approximation to the evolution operator by a product of $\Theta(rM\|H\|_{\rm max}/\herr)$ exponentials
of $U_l$.

In step 4, we divide the interval $[0,t]$ into segments of length $\Theta(1/(M\|H\|_{\rm max}))$.
The number of exponentials of $U_l$ in each such segment is therefore
\begin{equation}
\label{eq:m}
m = \Theta(1/s) = \Theta \left( M^2(\|H\|^2+\|H'\|) t/(\herr\varepsilon) \right) .
\end{equation}
Each exponential of $U_l$ can be implemented using an ancillary control qubit, as described above,
with failure probability $O(s)$.
Since the number of ancillary control qubits in each segment is $m = \Theta(1/s)$,
the overall probability of failure within each segment is $O(1)$.
It is simple to obtain the probability of failure less than $1/4$ with an appropriate choice of multiplicative constants.

As explained above, rather than controlling each $U_l$ on a single ancilla qubit, we use the simulation scheme of Fig.\ \ref{fig:unco} to implement the same operation. This allows us to compress the state of the ancillas~\cite{Berry12}.
There are $\Theta(tM\|H\|_{\rm max})$ segments, and to bound the total error by $O(\varepsilon)$,
the error from the compression of each segment must be bounded by $\varepsilon_1 = \Theta(\varepsilon/(tM\|H\|_{\rm max}))$.
This implies that the maximum Hamming weight per segment, or the number of controlled operations, can be bounded by $k_1$, with
\begin{equation}
k_1 = O\left( \frac{\log(1/\varepsilon_1)}{\log\log(1/\varepsilon_1)} \right) = O\left( \frac{\log(td\|H\|_{\rm max}/\varepsilon)}{\log\log(td\|H\|_{\rm max}/\varepsilon)} \right).
\end{equation}
The number of calls to the oracle for $H$ is then $O(td^2\|H\|_{\rm max} k_1)$.
Since $\|H\|_{\rm max} \le \|H\|$, the number of oracle calls is upper bounded by Eq.~\eqref{eq:addgat}
(with the double-log factors omitted for simplicity).

For the number of additional gates, there are some modifications from~\cite{Berry12} needed.
In the result given for the number of additional gates in Theorem 1 of~\cite{Berry12},
the first term ``$TG\log(T)$'' is for complexity due to a driving Hamiltonian, where $T$ refers
to the total evolution time and $G$ is the number of elementary gates to approximate the driving Hamiltonian evolution.
In the present case, there is no driving Hamiltonian, and that term can be omitted.
The other change is that here the time intervals are divided into segments of length $\Theta(1/(M\|H\|_{\rm max}))$, rather than of length $1/4$.
Therefore, in each segment the error parameter $\varepsilon_0$ ($\varepsilon$ in~\cite{Berry12}) is now required to be $\Theta(\varepsilon/(mtM \|H\|_{\rm max})$.
We take
$k_0 = O\left( \log(1/\varepsilon_0) \right)$,
and the number of additional gates is given by
\begin{equation}
O(tM\|H\|_{\rm max} k_0 [(\log m)^2+\log m \log\log(1/\varepsilon_0)]).
\end{equation}
We consider scaling for large $\|H\|_{\rm max}/\herr$, so
$1/\varepsilon_0 = O(m^2)$, and this becomes $O(td^2\|H\|_{\rm max} (\log m)^3)$.

Next, we place a bound on the value of $\herr$ needed.
If the discretized Hamiltonian is $\tilde H =\gamma \sum_l U_l$, we have an error bounded by
$\|H-\tilde H\|\le \sqrt{d} \|H-\tilde H\|_{\rm max}\le \sqrt{d}\herr$.
To ensure an error in the simulation bounded by $O(\varepsilon)$, we can therefore take $\herr=\Theta(\varepsilon/t\sqrt{d})$.
Substituting this expression, and the expression for $m$ in Eq.~\eqref{eq:m}, then gives the number of gates upper bounded by Eq.~\eqref{eq:addgat}.

So far we considered the average complexity, rather than an upper bound on the complexity.
In order to place an upper bound on the complexity, Refs.~\cite{CleveG+2009,Berry12} used the Markov bound and obtained a multiplying factor of $1/\varepsilon$.
That factor would lead to scaling that is polynomial in $1/\varepsilon$, but it is unnecessary.
Using the Chernoff bound, it can be shown that we can account for an upper bound on the complexity
by adding $\log(1/\varepsilon)$ to the number of segment attempts (see Appendix).
Therefore, the result changes so that $\log(1/\varepsilon)$ is added to $d^2\tau$ in Eq.~\eqref{eq:addgat}.
\end{proof}

Efficient methods for Hamiltonian simulation are important for studying physical quantum systems, as well
as a basis for other quantum algorithms. A drawback of previous simulation methods is that they have scaling
that is polynomial in the inverse error. In this paper, we provided a method that is polynomial in the logarithm
of the inverse error and the rate of change of the Hamiltonian, which is relevant for those cases in which the
Hamiltonians vary rapidly.

\acknowledgements
We acknowledge valuable discussions with A. M. Childs, S. Gharibian, and N. Wiebe.
Research supported by ARC grant FT100100761, Canada's NSERC, CIFAR, MITACS, and the U.S. ARO.
RDS acknowledges support from the Laboratory Directed Research and Development Program
at Los Alamos National Laboratory.

\appendix
\section{Appendix: Improved analysis of procedure to correct faults}

Here we show that the number of queries can be capped by 
\begin{align}\label{eq:query-bound}
O\left(\left[d^2\tau\frac{\log(d\tau/\varepsilon)}{\log\log(d\tau/\varepsilon)}+\log(1/\varepsilon)\log(d\tau/\varepsilon)\right]\log^*n\right),
\end{align}
and the number of additional gates can be capped by
\begin{align}\label{eq:gate-bound}
O\left([d^2 \tau+\log(1/\varepsilon)] \log^3[d(\tau+\tau')/\varepsilon]n^c\right),
\end{align}
with probability $O(\varepsilon)$ of exceeding these caps.
We can assume that the algorithm does not exceed the cap, with cases where more queries or gates would be required regarded as an error.
As that only occurs with probability $O(\varepsilon)$, its contribution to the total error is $O(\varepsilon)$.
For Theorem 1 in the main text we simply present Eq.~\eqref{eq:gate-bound} for the sum of the two bounds above.

In Eq.~\eqref{eq:query-bound}, $\log^*$ is the iterated logarithm.
That is, it is the number of times the logarithm must be performed to reach a constant.
This factor comes from using the method for decomposing a sparse Hamiltonian into 1-sparse components from Ref.~\cite{Berry2007}, which scales as $\log^* n$.
There are no other factors dependent on $n$ that affect the number of queries.
For the number of additional gates, the exact dependence on $n$ will depend on how the controlled operations are performed.
It will be polynomial in $n$, which is why the scaling is given as $n^c$.

In~\cite{CleveG+2009,Berry12}, similar expressions were obtained, but containing factors of $1/\varepsilon$.
The context of those works (quantum query complexity) permitted $\varepsilon$ to be essentially treated as a constant, which is not so in the present context, where the scaling with respect to $\varepsilon$ is relevant.
In the analysis of~\cite{CleveG+2009,Berry12}, the $1/\varepsilon$ factors arise as a result of applying the Markov bound to the expected complexity in three places:
\begin{enumerate}
\item
the number of attempted segment computations,
\item
the additional cost of correcting errors, and
\item
the number of gates for the recursive measurement.
\end{enumerate}
Items 1 and 2 are discussed in Sec.~4.1 of Ref.~\cite{Berry12}, and item 3 is discussed in Sec.~6 of Ref.~\cite{Berry12}.
Below, in Secs.~A1 to A3, we use different techniques to obtain tighter bounds with respect to $\varepsilon$ in these respective places.
This analysis is performed assuming that the measurement results are obtained with the probabilities that arise in an uncompressed scheme.
This is valid because the error due to compression is already accounted for.
The compression per segment is set so that the error due to the compression is $O(\varepsilon/N)$, where $N$ (defined in Sec.~A1) upper bounds 
the number of attempted segment computations performed (including corrections of faults).
Note that this allowable error due to compression is different than in the main text.
The allowable error can be set to this value without changing the complexity (see Sec.~A2).

\subsection{A1. The number of attempted segment computations}
\label{sec:segs}
The number of segments to compute is $T := O(d^2 \tau)$.
Each attempt to compute a segment succeeds with probability at least 
$\frac{3}{4}$.
When a segment attempt fails (in the sense of having faults), then this can be regarded as a step backwards in a random walk.
An attempt is made to undo the failed attempt---including the faults that occurred within it.
This attempt to undo also succeeds with probability at least $\frac{3}{4}$.
If the undo succeeds then the random walk has returned to the initial position, and the segment can be attempted again.
If the undo fails then the random walk takes another step backwards, and an attempt must be made to undo the previous undo along with its faults.

We use the terminology that a ``segment attempt'' is the initial attempt to perform the segment, as well as any attempt to undo failed attempts.
A ``segment'' is a complete successful segment, which may be composed of many segment attempts.
The process of achieving a segment corresponds to a biased random walk composed of segment attempts,
with probability of forward steps at least $\frac{3}{4}$ (and otherwise backward steps).
When this walk advances one step, then a segment is successfully achieved.
The query costs associated with the the segment attempts are not all the same (which we will account for in Sec.~A2).
Concatenating the random walks for all $T$ segments, there is an overall random walk that advances $T$ steps,
and the parts of the walk corresponding to different segments do not need to be distinguished in the analysis of the walk.

To complete $T$ segments, it suffices for the number of successful segment attempts minus the number of failed segment attempts to be at least $T$.
After $2T$ segment attempts, the \textit{expected} position of the random walk is at least $T$; however, the probability of failing to reach position $T$ might not be small.
We will first show that the number of segment attempts can be bounded by 
\begin{align}\label{eq:cap-number-segments}
N = O(T + \log(1/\varepsilon))
\end{align}
with probability of failing to reach position $T$ at most 
$\varepsilon$.

To show this, assume that we cap the number of segment attempts at 
$N = 2T+\lambda$, for some parameter $\lambda$ (that will be determined later).
Let $X$ be the random variable corresponding to the number of successful segment attempts (among all of the $N$ trials).
Note that the expected value of $X$ (which we denote by $\mu$) is 
$\mu \ge \frac{3}{4}N = \frac{3}{2}T + \frac{3}{4}\lambda$.
Failure to reach position $T$ implies  
\begin{align}
X 
& <  \textstyle{\frac{1}{2}}(N+T) \nn
& =  \textstyle{\frac{3}{2}}T + \textstyle{\frac{1}{2}}\lambda \nn
& =  \left(1 - \textstyle{\frac{\lambda}{6T+3\lambda}}\right)\left(\textstyle{\frac{3}{2}}T + \textstyle{\frac{3}{4}}\lambda\right)
\end{align}
(where the last equality is an algebraic identity).
Therefore, setting $\delta = {\lambda}/(6T+3\lambda)$ in the Chernoff bound \cite{Chernoff},
\begin{align}\label{eq:chernoff}
{\rm Pr}[X < (1-\delta)\mu] < \exp(-\delta^2\mu/2),
\end{align}
implies that the failure probability is bounded above by 
\begin{align}
\exp\left[-\textstyle{\frac 12\left(\frac{\lambda}{6T+3\lambda}\right)^2}
\textstyle{\left({\frac{3}{2}T + \frac{3}{4}\lambda}\right)}\right] 
= \exp\left(-\textstyle{\frac{1}{24}\frac{\lambda^2}{2T+\lambda}}\right).
\label{eq:random-walk-analysis}\end{align}
This is at most $\varepsilon$ by setting  
$\lambda = O(T + \log(1/\varepsilon))$.
It follows that the required number of segment attempts is bounded by the expression in Eq.~(\ref{eq:cap-number-segments}) up to error probability at most $\varepsilon$.

\subsection{A2. The additional cost of correcting errors}
\label{sec:corr}
Next, we account for the query costs associated with each segment attempt.
There are two sources associated with the query cost of each segment attempt.
One is the basic query cost associated with each segment attempt including the initial attempt.
This cost is $O(Nk_1)$; using $k_1$ from Eq.~(10) of the main text, this cost is $O(N \log(\segs/\varepsilon)/\log\log(\segs/\varepsilon))$.
Here we also need to take account of the fact that we are using $N$ segment attempts, so we only allow error of $O(\varepsilon/N)$ per segment.
Using this to determine $k_1$, we still obtain $k_1=O(\log(\segs/\varepsilon)/\log\log(\segs/\varepsilon))$.
This means that modifying the allowed error per segment in this way does not alter the complexity.

Another, additional query cost, arises when faults occur in segment attempts.
If the first attempt at a segment fails, resulting in $f_1$ faults, then the subsequent attempt to undo that failure makes $2f_1$ additional queries at locations associated with the positions of these faults.
This is because a fault corresponds to applying the operation $e^{i U_l \pi/2}$, and an exact implementation of this operation can be performed using $U_l$ twice and other gates \cite{CleveG+2009}.
If this segment attempt also fails, resulting in $f_2$ new faults, then the next attempt makes $2(f_1 + f_2)$ additional queries.
If that segment attempt is successful, then the next segment attempt (for the first failed segment) needs $2f_1$ additional queries.
In this manner, every fault that occurs contributes to the query cost of \textit{every} segment attempt that is needed to undo the failed segment where that fault occurred.

There are $\segs$ segments that need to be successfully computed. 
For each such segment, we let $R$ and $E$ be the random variables
associated with the number of attempts and the total 
number of faults, respectively.
For a single segment \emph{attempt}, $f$ is the random variable for the number of faults.
The probability distribution for $f$ is denoted $q(f)$, and is the binomial
\begin{equation}
q(f) = \binom{m}{f} \left( \frac 1{4m} \right)^f \left( 1- \frac 1{4m} \right)^{m-f}.
\end{equation}
This is because there are $m$ ancillas and each has a failure probability of $1/(4m)$.
Note that $q(0) \ge 3/4$ is the probability of success of any attempt.
In general we have
\begin{equation}
\label{eq:prbnd}
q(f) \le e^{-1/4} \frac{1}{f!} \left( \frac 1{4-1/m} \right)^{f}.
\end{equation}

When an attempt fails, we then require two other attempts, first to undo the failed attempt, then to redo it.
Both also have probability $q(0)$ to succeed.
If there are $R$ attempts used to successfully compute a segment, then there are $n$ failed attempts and $n+1$ successful attempts, with $R=2n+1$.
Omitting the last attempt (which must be successful), there are $2n$ attempts, and no more than $2^{2n}=4^n$ possible sequences of successes and failures.
Denoting the exact number of sequences by $C_n$, the probability distribution for $n\ge 0$ is then
$\Pr(n)=C_{n} q(0)^{n+1} [1-q(0)]^{n}$.

For a given $R$, the total number of faults is $E=f_1 + \ldots  + f_n$, with $f_j \ge 1$.
The total number of additional queries, $Q$, is then upper bounded by $2(2n-1)E$.
This is because each fault need be corrected no more than $2n-1$ times, and each correction uses $2$ queries.
The probability distribution for the numbers of faults $f_1, \ldots, f_{n} \ge 1$ is $C_{n} q(0)^{n+1}q(f_1) \ldots q(f_{n})$.
Then,
\begin{align}
\mE [e^{tQ}]   
& \le q(0) + \sum_{n \ge 1} C_{n} q(0)^{n+1} \nn
& \quad \times \sum_{f_1,\ldots,f_{n} =1}^m q(f_1) \ldots q(f_{n}) e^{2t(2n-1)(f_1 + \ldots + f_{n})} \nn
& = \sum_{n \ge 0} C_{n} q(0)^{n+1} \left[ \sum_{f = 1}^m q(f ) e^{2t (2n-1) f } \right]^{n} .
\end{align}

The right-hand side diverges if $n$ is allowed to be unbounded.
However, we can bound $n$ by $n_M$ of order $\log (T/\varepsilon)$ using the Chernoff bound,
incurring an error probability of order $\varepsilon/T$ (this follows from substituting $1$ for $T$ and $\varepsilon/T$ for $\varepsilon$ in Eq.~\eqref{eq:cap-number-segments}).
Taking this error probability to be exactly $\varepsilon/T$ for simplicity,
we define $x:=1/(1-\varepsilon/T)$.
Given that the whole simulation succeeded (i.e., $n \le n_M$)
we redefine the probabilities as, for $n \le n_M$,
$\bar \Pr(n) = x \Pr(n)$, and for $n > n_M$, $\bar \Pr(n) = 0$. This implies that, given that the simulation succeeded,
\begin{align}
\mE [e^{tQ}]/x   
& = \sum_{n= 0}^{n_M} C_{n} q(0)^{n+1} \left[ \sum_{f =1}^m q(f ) e^{2t (2n-1) f } \right]^{n} =: g(t) .
\end{align}
The function $g(t)$ is analytic and monotonically increases with $t$.
Because $\mE [1]=1$, we have $g(0)=1/x$.

The derivative of the function $g(t)$ satisfies
\begin{align}
\dot g (t) &=  \sum_{n = 1}^{n_M} C_{n} q(0)^{n+1} n \left[ \sum_{f =1}^m q(f ) e^{2t (2n-1) f } \right]^{n-1} \nn &\quad \times
 \sum_{f' = 1}^m q(f' )  2(2n-1) f' e^{2t (2n-1) f '}.
\end{align}
Suppose we are in a regime where $t \le  1/(4 n_M)$.
Using Eq.~\eqref{eq:prbnd}, the sum over $f '$ is upper bounded by a constant as
\begin{align}
& \sum_{f' = 1}^m q(f' )  f' e^{2t (2n-1)f '} \le \sum_{f' = 1}^m q(f' )  f' e^{f '}\nn
& \le \sum_{f' \ge 1} e^{-1/4} [1/(4-1/m)]^{f'}  f' e^{f '}/f'! \nn
& = \sum_{f' \ge 1} e^{-1/4} [e/(4-1/m)]^{f'}  f'  /f'! \nn
& = e^{[e/(4-1/m)]-1/4} [e/(4-1/m)] < 2 .
\end{align}

Using this gives the bound on the derivative of $g(t)$ as
\begin{align}
\dot g (t) \le 8 \sum_{n = 1}^{n_M} C_{n} q(0)^{n+1} n^2   \left[ \sum_{f=1}^m q(f ) e^{ 2t(2n-1)f } \right]^{n-1}\; .
\end{align}
Next, if $t \le 1/(40 n_M)$ and $m\ge 35$, then
\begin{align}
\sum_{f =1}^m q(f ) e^{f/10} & \le \sum_{f \ge 1} e^{-1/4} [e^{1/10}/(4-1/m)]^{f} /f! \nn
&= e^{-1/4}(e^{e^{1/10}/(4-1/m)}-1) \nn
&=\beta < 1/4 .
\end{align}
Note that we can take $m$ to be at least a constant without changing any scaling.
Then, using $C_n \le 4^n$, and noting that $4 \beta<1$, we have
\begin{align}
\dot g (t) &\le 8 q^2(0) \sum_{n=1}^{\infty} n^2 q(0)^{n-1} (4\beta)^{n-1} \nn
&\le 8 q^2(0) \sum_{n=1}^{\infty} n (n+1) q(0)^{n-1} (4\beta)^{n-1}  \nn
& =  16 q^2(0) \frac 1 {[1- 4 q(0) \beta]^3}.
\end{align}
That is, $\dot g (t) \le \alpha$ for some constant $\alpha$.
Therefore, we obtain
\begin{align}
g(t)&=1/x+\int_0^t \dot g(t') dt' \nn
&\le 1/x + \alpha t
\end{align}
for $t \le 1/(40 n_M)$.
 
We now let $Z$ be the random variable that sums all the $Q$ for the $\segs$ segments.
Recall that these random variables are taken to be postselected on each segment taking no more than $O(\log (\segs/\varepsilon))$ attempts to compute.
Using $\mE [e^{tQ}]=xg(t)$, we have
\begin{equation}
\mE [e^{tQ}] \le 1+x\alpha t \le e^{x\alpha t}.
\end{equation}
Using the Markov inequality, if $t>0$ and the simulation succeeded,
\begin{align}
\label{eq:chernoff}
\Pr[Z \ge \lambda T]
&\le  \frac{\mathbb{E} [e^{tZ}]} {e^{t \lambda \segs }}  \nn
& \le \left( \frac{\mathbb{E}[e^{tQ}]} {e^{t \lambda  }} \right)^\segs  \nn
& \le \left( \frac{e^{x \alpha  t}} {e^{t \lambda  }} \right)^\segs .
\end{align}
Then, for error probability of order $\varepsilon$, it suffices to satisfy
\begin{align}
t(\lambda-x\alpha)\ge \log(1/\varepsilon)/\segs .
\end{align}
That is, it suffices that
\begin{align}
\lambda \ge \left( \frac {\log(1/\varepsilon)} {t\segs} + x\alpha \right).
\end{align}
Therefore we can choose $\lambda$ such that
\begin{align}
\lambda \segs = O(\segs + \log(1/\varepsilon)/t)  .
\end{align}
Recall that above we have required $t \le 1/(40 n_M)$ for $n_M=O(\log(1/\varepsilon))$.
Taking $t$ as large as is permitted, we obtain
\begin{align}
\label{eq:corcor}
\lambda \segs = O(\segs + \log(1/\varepsilon)  \log(\segs/\varepsilon)) .
\end{align}

For the number of queries, we need to first take account of the fact that the number of segment attempts is larger than the number of segments.
Then we need to take account of the additional queries needed to correct the faults.
For the first correction, replacing $\segs$ with $N$ in the expression $Nk_1$
for the number of segment attempts yields the number of queries as,
\begin{equation}
O\left([\segs+\log(1/\varepsilon)] \frac{\log(\segs/\varepsilon)}{\log\log(\segs/\varepsilon)}\right).
\end{equation}
Taking the maximum of this expression and Eq.~\eqref{eq:corcor} yields Eq.~\eqref{eq:query-bound}.
The additional $\log^* n$ takes account of the number of queries needed for the decomposition into 1-sparse Hamiltonians in the method from Ref.~\cite{Berry2007}.

\subsection{A3. The number of gates for the recursive measurement}
\label{sec:fail}
Next we consider the number of gates used for the recursive measurement.
This depends on the number of faults detected.
If the number of faults detected is $f$, then the number of steps in the recursive procedure is
$1+2f\log m$.
For all $N$ segment attempts, we detect $O(N)$ faults, resulting in the number of steps in the recursive procedures for all attempts being
$O(N\log m)$.

Using Eq.~\eqref{eq:cap-number-segments}, we can then modify Eq.~(11) of the main text by replacing $tM\|H\|_{\rm max}$ with
$tM\|H\|_{\rm max}+\log(1/\varepsilon)$.
That then yields Eq.~\eqref{eq:gate-bound} as the bound on the number of gates.
\end{document}